\documentclass[journal,comsoc,twoside]{IEEEtran}
    \usepackage{newtxtext,newtxmath}
    \ifCLASSOPTIONcompsoc
      \usepackage[caption=false,font=footnotesize,labelfont=sf,textfont=sf]{subfig}
    \else
      \usepackage[caption=false,font=footnotesize,labelfont=small]{subfig}
    \fi
    \ifCLASSOPTIONcaptionsoff
      \usepackage[nomarkers]{endfloat}
     \let\MYoriglatexcaption\caption
     \renewcommand{\caption}[2][\relax]{\MYoriglatexcaption[#2]{#2}}
    \fi
    
    \usepackage[utf8]{inputenc}
    \usepackage{ragged2e}
    \usepackage{units}
    \usepackage{graphicx}
    \usepackage{tabulary}
    \usepackage{booktabs}
    \usepackage{mathtools}  
    
    \usepackage{amsmath} 
    \usepackage{amsthm} 
    
    \usepackage{amssymb}
    \usepackage{amsfonts}
    \usepackage[noadjust]{cite}
    \usepackage{balance} 
    \usepackage{tabularx}
    \usepackage{comment}    
    \usepackage{setspace}
    \usepackage{xspace}
    \usepackage{acronym}
    \usepackage{multirow}
    \usepackage{upgreek} 
    \usepackage{nicefrac} 
    \usepackage{ifthen}
    \usepackage{csvsimple}
    \usepackage{hyperref}
    \usepackage{flexisym} 
    \usepackage{fp} 
    \usepackage{caption}
    \usepackage{titlecaps}
    \usepackage{optidef}
    \usepackage[thinc]{esdiff} 
    
    \theoremstyle{definition}

    \DeclareMathOperator{\Expectation}{\mathbb{E}\xspace}
    \DeclareMathOperator{\Probability}{\mathbb{P}\xspace}
    \usepackage{xargs} 
    \usepackage{xcolor}
    
    \acrodef{IoT}{Internet of Things}
    \acrodef{RV}{random variable}
    \acrodef{RAO}{random access opportunity}
    \acrodef{CIoT}{cellular IoT}
    \acrodef{LTE}{Long Term Evolution}
    \acrodef{LTE-A}{LTE-Advanced}
    \acrodef{LTE-M}{LTE-MTC}
    \acrodef{NB-IoT}{Narrowband \ac{IoT}}
    \acrodef{NR}{new radio}
    \acrodef{mMTC}{massive machine-type communications}
    \acrodef{MIMO}{multiple-input multiple output}
    \acrodef{mMIMO}{massive \acl{MIMO}}
    \acrodef{NOMA}{non-orthogonal multiple access}
    \acrodef{MAC}{medium access control}
    \acrodef{MAD}{maximum average distance}
    \acrodef{ZC}{Zadoff–Chu}
    \acrodef{FDD}{frequency division duplexing}
    \acrodef{BS}{base station}
    \acrodef{5G-NR}{5G new radio}
    \acrodef{PRACH}{Physical Random Access Channel}
    \acrodef{RACH}{Random Access Channel}    \acrodef{RA}{random access}

    \newcommand{\RevOne}[1]{#1\xspace} 
    
    \newcommand{\MyCeRA}{OptCeRA\xspace}
    \newcommand{\CeRA}{CeRA\xspace}
    \newcommand{\CeRATwo}{Multipreamble RA\xspace}
    
    \newcommand{\nth}[1]{\textit{#1}th}
    \newcommand{\etal}{\mbox{\emph{et al.}}\xspace}

    \newcommand{\myiff}{if and only if\xspace} %
    \newcommand{\iid}{i.i.d.\xspace}  
    \newcommand{\Aqn}[2]{\mathcal{A}_{#1}^{#2}\xspace}
    \newcommand{\Aq}[1]{\mathcal{A}_{#1}\xspace}

    \newcommand{\xij}[2]{\pi_{#1}(#2)}
    
    \newcommand{\gformat}[1]{\mathcal{#1}}
    \newcommand{\HSymbol}{H}
    \newcommand{\V}{\gformat{X}}
    \newcommand{\E}{\gformat{E}}
    \newcommand{\Y}{\gformat{Y}}

    \newcommand{\HGraph}[1]{\HSymbol(\V,\E)}
    \newcommand{\Hypergraph}[2]{\HSymbol_{\code}(\V,\E)} 
    
    \newcommand{\IGraph}[2]{\IHSymbol{#1}=(#1,#2)}
    \newcommand{\IHSymbol}[1]{\HSymbol_{#1}\xspace}
    
    \newcommand{\hd}[2]{H(#1,#2)\xspace}
    \newcommand{\ahd}[1]{d_H(#1)\xspace}
    
    \newcommand{\Bknp}[3]{\mathbb{B}_{#1}(#2,#3)\xspace}
    \newcommand{\piij}{\xij{i}{j}\xspace}
    
    \newcommand{\code}{\mathcal{C}\xspace}
    
    \newcommand{\madcode}[2]{$(#1,#2{\cdot}k)_{#2}$} 
    \newcommand{\codeacr}{$(n,\M)_{q}$} 
    \newcommand{\alphabetElement}{alphabet element\xspace}
    \newcommand{\coordinate}{codeword coordinate\xspace}
    \newcommand{\qaryRepresentation}{\kb = \as_{\nc}\ql^{\nc-1}+\as_{\nc-1}\ql^{\nc-2}+\cdots+\as_{1}\ql^{0}\xspace}
    \newcommand{\aSequence}{\mathbf{\as}=(\as_{\nc},\as_{\nc-1},\cdots,\as_{1})\xspace}
    \newcommand{\cSequence}{\mathbf{\elementC}=(\elementC_{n},\elementC_{n-1},{\cdots},\elementC_{1})\xspace}
    \newcommand{\cCodeword}{\mathbf{\elementC}=\elementC_{n}\elementC_{n-1}{\cdots}\elementC_{1}\xspace}
    \newcommand{\eSequence}{\sse_1,\sse_2,{\cdots},\sse_n\xspace}
    
    \newcommand{\POne}{\mathcal{P}\mathit{1}\xspace}
    \newcommand{\PTwo}{\mathcal{P}\mathit{2}\xspace}
    \newcommand{\MADCproblem}{\textit{\ac{MAD} code problem}\xspace}
    
    \newcommand{\codeambiguityproblem}{\textit{code ambiguity problem}\xspace} 
    
    \newcommand{\resourceefficiency}{grant utilization\xspace}

    \newcommand{\device}{device\xspace}
    \newcommand{\codeword}{codeword\xspace}
    
    \newcommand{\validcw}{\RevOne{inferred} valid codeword\xspace}
    
    \newcommand{\coverededge}{\textit{covered hyperedge}\xspace}
    \newcommand{\coveredcw}{\textit{covered \codeword}\xspace}
    \newcommand{\nc}{n\xspace}

    \newcommand{\subframe}{subframe\xspace} \newcommand{\superframe}{superframe\xspace} 
    \newcommand{\rasubframe}{\ac{RA} {\subframe}\xspace}
    \newcommand{\rasuperframe}{\ac{RA} {\superframe}\xspace} 
    \newcommand{\schemetype}{code-expanded\xspace} 
    
    \newcommand{\rvX}{X\xspace}

    \newcommand{\rvU}{\eta\xspace}
    
    \newcommand{\rvC}{C\xspace}
    \newcommand{\rvV}{V\xspace}
    \newcommand{\rvG}{G\xspace}
    \newcommand{\rvA}{U\xspace}

    \newcommand{\I}{\mathcal{I}\xspace}
    \newcommand{\J}{\mathcal{J}\xspace}
    \newcommand{\T}{\mathcal{T}\xspace}
    \newcommand{\Q}{\mathcal{S}\xspace}
    \newcommand{\Z}{\mathbb{Z}\xspace}
    
    \newcommand{\qary}{$q$-ary\xspace}

    \newcommand{\elementC}{c\xspace} 
    \newcommand{\elementVS}{x\xspace} 

    \newcommand{\sse}{e\xspace}
    \newcommand{\se}{E\xspace}

    \newcommand{\K}{K\xspace}       
    \newcommand{\BigL}{n\xspace}    
    \newcommand{\M}{M\xspace}       %
    \newcommand{\N}{N\xspace}       
    \newcommand{\R}{R\xspace}

    \newcommand{\rl}{a\xspace}
    \newcommand{\rs}{r\xspace}
    \newcommand{\as}{a\xspace}
    \newcommand{\ql}{q\xspace}
    \newcommand{\ka}{k\xspace}

    \newcommand{\NC}{N_{C}\xspace}
    \newcommand{\NS}{N_{S}\xspace}
        
    \newcommand{\kb}{t\xspace}
    \newcommand{\Nb}{K\xspace}
    \newcommand{\pb}{p\xspace}

    \newcommand{\Palloc}{P_{alloc}\xspace}

    \newcommand{\network}{network\xspace}

    \DeclareMathSizes{10}{9}{7}{5}
    
    \newboolean{long}   

    \def\PRatioWidthFigures{0.32\linewidth} 
    \def\PIHSpaceFigures{-0.0em}
    \def\PHSpaceFigures{-0.0em} 

    \def\cBallsReceived{green}
    
    \def\cBallsValid{blue}
    \def\cBallsCodebook{white}
    \def\cNonInCode{black}
    \def\cBallsUndetectedPreambles{gray}
    \def\cBallsPreambles{red}

    \newtheorem{definition}{Definition}
    \newtheorem{theorem}{Theorem}
    
    \newtheorem{lemma}{Lemma}

    \newcommand{\fOneWR}{2\left(1-\nicefrac{1}{r}\right)^{\Nb}}

    \newcommand{\fTwoWR}{\left(1-\nicefrac{2}{r}+\nicefrac{1}{\M}\right)^{\Nb}}
    
    \newcommand{\fExpectationVNTwo}{\M {\cdot} \left[1-\left[\fOneWR-\fTwoWR\right]\right]\xspace}

    \usepackage{letltxmacro}
    \LetLtxMacro{\oldsqrt}{\sqrt}
    \renewcommand{\sqrt}[2][\mkern8mu]{\mkern-4mu\mathop{\oldsqrt[#1]{#2}}}

\begin{document}
%
\author{
\IEEEauthorblockN{Carlos~A.~Astudillo,~\IEEEmembership{Student Member,~IEEE},
Ekram~Hossain,~\IEEEmembership{Fellow,~IEEE}, and Nelson~L.~S.~da~Fonseca,~\IEEEmembership{Senior Member,~IEEE}%
\vspace{-5mm}
\thanks{
This work was supported by the São Paulo Research Foundation (FAPESP) under grant number 15/24494-8 and 19/22065-3, and the Government of Canada through the Emerging Leaders in the Americas Program (ELAP). 
}%
\thanks{C.~A.~Astudillo and N.~L.~S~da~Fonseca are with the Institute of Computing, University of Campinas 13083-852, Brazil. %
E.~Hossain is with the Department of Electrical and Computer Engineering, University of Manitoba, Winnipeg, MB R3T 5V6, Canada
(e-mails: castudillo@lrc.ic.unicamp.br, ekram.hossain@umanitoba.ca, nfonseca@ic.unicamp.br).}%
} }

\title{\huge Random Access Based on Maximum Average Distance Code for Massive MTC in Cellular IoT Networks}%

\maketitle
\begin{abstract}
    Code-expanded Random Access (CeRA) is a promising technique for supporting \acl{mMTC} in cellular \acs{IoT} networks. 
    However, its potentiality is limited by code ambiguity, which results from the inference of a larger number of \codeword{s} than those actually transmitted.
    In this letter, we propose a \ac{RA} scheme to alleviate this problem by allowing  devices to select the preambles to be transmitted considering a \textit{q}-ary code with \acl{MAD}. 
    Moreover, a CeRA decoding approach based on hypergraphs is proposed and an analytical model \RevOne{is} derived. 
    Numerical results show that the proposed scheme significantly increases the probability of successful channel access as well as resource utilization.
\end{abstract}
\acresetall
\vspace{-2mm}
%
\IEEEpeerreviewmaketitle
\bstctlcite{IEEEexample:BSTcontrol}
\section{Introduction}
\label{sec:Introduction}
    \IEEEPARstart{T}{he}
    main challenge in supporting \ac{mMTC} use case category in 5G networks is the handling of a massive number of devices in the \acf{RA} procedure. 5G networks rely on New Radio (NR) and \ac{CIoT} network technologies, such as \acs{LTE-M} and \acs{NB-IoT}, for supporting \ac{mMTC}. These technologies employ \acs{RA} procedures based on grants, which require the transmission of an orthogonal preamble sequence before the reception of a grant for data transmission.
    One of the limiting factors is the reduced number of \ac{RA} contention resources (e.g. up to $48$ orthogonal preambles in NB-IoT or $64$ in LTE-M/NR), which can lead to a large number of transmission collisions. Such collisions decreases the probability of successful channel access and reduces the efficiency in the usage of available resources.
    
    Several improvements to grant-based \ac{RA} procedures have been proposed, including overload control, \ac{NOMA} transmissions, and collision detection  and resolution.
    Some proposals adopt the transmission of more than one preamble \cite{Mostafa_8647963,Thomsen_TETT2013,Vural_CL20187-7908954} in order to overcome the  limitation of \ac{RA} resources as well as reduce the chances that two or more devices select the same RA contention resource. 
    One alternative is to aggregate two preambles with different \ac{ZC} root sequences in the same \rasubframe \cite{Mostafa_8647963}, which may require hardware modification and can considerably increase the  multiple-access interference due to the loss of orthogonality. Moreover, this cannot be applied to the NB-IoT preamble, which is based on a single-tone waveform combined with frequency hopping.

    \begin{figure}[!t]
    \vspace{-4mm}
    \centering
    \subfloat[\CeRA scheme]
    {\label{fig:codeambiguity_cera}\hspace{-0mm}
    \includegraphics[width=0.476\linewidth]{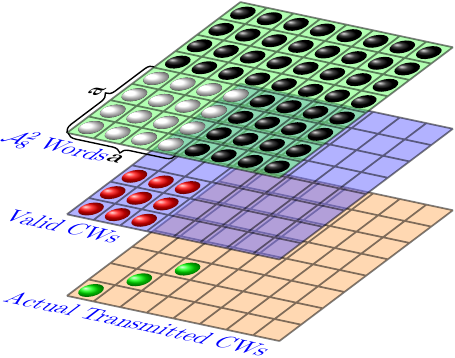}
    }
    \subfloat[\MyCeRA scheme]
    {\label{fig:codeambiguity_mycera}\hspace{-0mm}
    \includegraphics[width=0.43316\linewidth]{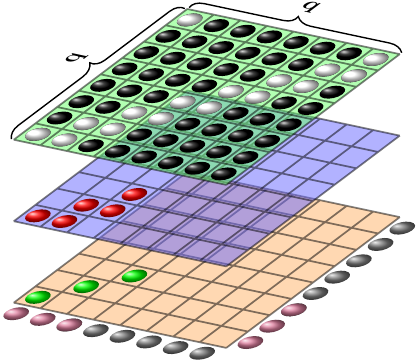}
    }
    \caption{Illustration of \CeRA decoding and the \codeambiguityproblem. 
    $2$ \rasubframe{s} per \superframe, $8$ preambles available per \subframe, $16$ \codeword{s} (\cBallsCodebook \xspace balls), and $3$ transmitted \codeword{s} (\cBallsReceived \xspace balls); 
    \cNonInCode \xspace balls denote words in $\Aqn{8}{2}$ not in the code;
    \cBallsUndetectedPreambles \xspace and \cBallsPreambles \xspace balls represent, respectively, the idle and detected preambles;
    and \cBallsValid \xspace balls denote the valid \codeword{s}.
    }
    \vspace{-4mm}
    \label{fig:schematics_motivation}
    \end{figure}
    
    A more attractive alternative is \schemetype \ac{RA} (CeRA) \cite{Thomsen_TETT2013,Vural_CL20187-7908954}, in which each device performs random access by using consecutive transmissions of randomly selected preambles over $\BigL$ \rasubframe{s}, to form an \rasuperframe. 
    At every \rasubframe, a preamble is randomly selected from the set of $\rl$ allowed preambles.
    This sequence of preambles over the \superframe is interpreted as a \codeword of length $\BigL$. Uplink resources for data transmission are allocated on the basis of the \codeword{s} inferred by the \ac{BS} in the \superframe.
    
    However, in this scheme, individual \codeword{s} are not received by the \ac{BS}, where the preamble transmission from several \device{s} at each \subframe is identified. Consequently, a larger number of \codeword{s} can be considered valid than those actually transmitted in a \superframe. 
    \RevOne{For example, consider a \CeRA scheme with two \rasubframe per \superframe, four preambles allowed in a subframe among the eight available preambles, and $16$ \codeword{s} as illustrated in Fig. \ref{fig:codeambiguity_cera}.
    The possible codewords for the devices to perform \ac{RA} are represented by \cBallsCodebook \xspace balls.
    Consider that $3$ different \codeword{s} (\cBallsReceived \xspace balls) are transmitted by the devices in a \superframe. The \ac{BS} detects the preambles transmitted at each \subframe of a \superframe (\cBallsPreambles \xspace balls) and generates the set of \validcw{s} (\cBallsValid \xspace balls).
    Note that all possible combinations of preambles received in different \subframe{s} within the \superframe are regarded as valid.} %
    This problem is known as the \codeambiguityproblem \cite{Vural_CL20187-7908954} and it impacts greatly on the utilization of uplink resources as well as on success probability under resource constraints.

    Only recently, various solutions have been proposed for  the \codeambiguityproblem \cite{Thomsen_TETT2013,Vural_CL20187-7908954,Jiang_WOCC2019,Jiang_TWC2019-8770292}. 
    It has been proposed that the \ac{MAC} layer adapts the number of preambles allowed in a \subframe considering the device load in order to reduce the number of possible combinations of preambles that can be inferred \cite{Vural_CL20187-7908954,Thomsen_TETT2013}. However, this strategy does not directly tackle the problem. On the other hand, solutions using \ac{mMIMO} \cite{Jiang_WOCC2019,Jiang_TWC2019-8770292} assume perfect channel estimation and spatial orthogonality which are hard to achieve in practice, especially in \ac{mMTC} scenarios. These solutions increase the system complexity and 
    cannot be used by \ac{CIoT} technologies since they employ low resolution \acs{MIMO} settings and either \acl{FDD} or half-duplex. 
    The \codeambiguityproblem must be addressed to exploit the true potentiality of \CeRA and enable its wide adoption for \acs{mMTC} in 5G and beyond networks.
    
    This paper proposes \MyCeRA, a \CeRA scheme based on \ac{MAC} coding to alleviate the \codeambiguityproblem.
    In the proposed scheme, the transmitted preambles are chosen on the basis of a new \qary \ac{MAD} code, which 
    \RevOne{minimizes the similarity between codewords}. 
    Moreover, this paper introduces an encoding procedure for a \ac{MAD} code, and a decoding approach for \CeRA schemes based on hypergraphs, as well as an analytical model for \RevOne{\CeRA schemes}. 
    \RevOne{Table \ref{tab:notation} shows the notation used in the paper.}
    
    \begin{table}[!t]
        \centering
        \small
        \RevOne{\caption{Notation}
        \label{tab:notation}
        \renewcommand{\arraystretch}{0.95}
        \resizebox{\columnwidth}{!}{%
        \begin{tabular}{l||l}
        \hline
        \bfseries Symbol & \bfseries Definition \\
        \hline\hline
        $\I$ & Index set of the codeword coordinates \\
        $\J$ & Index set of the code symbols\\
        $\T$ & Index set of the codewords \\
        $\Palloc$ & Resource allocation probability \\
        $P_{N}$ & Codeword non-colission probability\\
        $P_{S}$ & RA success probability\\
        $\rvU$ & Grant utilization \\
        $K$ & Number of contending devices\\
        $M$ & Number of available codewords \\
        $N$ & Number of words in $\Aqn{q}{n}$\\
        $R$ & Number of available resources per superframe\\
        $V$ & Number of \validcw{s}\\
        $a$ & Number of allowed preambles \\
        $k$ & Multiplicative factor of available preambles \\
        $n$ & Number of RA subframes per superframe \\
        $q$ & Number of available preambles \\
        $r$ & Number of preambles used in the code\\\hline
        \end{tabular}}}
        \vspace*{-0.8em}
    \end{table}

\section{Maximum Average Distance Codes and their Application to Code-Expanded RA}
\label{sec:proposed_solution}   
    This section introduces the \MyCeRA scheme to  tackle the \codeambiguityproblem. \MyCeRA optimizes the code used to perform  consecutive preamble transmissions.
    For some fixed value of $n$, when the number of available \codeword{s} is reduced, redundancy is introduced, imposing constraints to the transmitted messages (\codeword{s}) in a way that just a fraction of all possible received messages (all combinations of the received preambles at each \subframe) are  valid \codeword{s}. 
    By reducing the similarity between the \codeword{s} in the code, the number of \validcw{s} can be reduced as well. This reduction can be formulated as the maximization of the average Hamming distance between the \codeword{s} in a code.

    In contrast to the existing \CeRA schemes, in which the \codeword is formed by a sequence of preambles selected randomly from the $\rl$ allowed preambles, in the \MyCeRA scheme, the \codeword is first randomly selected from a code \RevOne{$\code$} with maximum average Hamming distance (a \acs{MAD} code), considering all the $q$ available preambles, \RevOne{i.e. $\code \subseteq \Aqn{q}{n}$, where $\Aqn{q}{}$ is a finite set of $q$ symbols (a $q$-ary alphabet)}. Then, a preamble is transmitted at the \nth{i} \subframe according to the \nth{i} coordinate of the chosen \codeword. 
    \RevOne{Using the \MyCeRA approach, the CeRA scheme in \cite{Vural_CL20187-7908954} becomes equivalent to randomly selecting a codeword from a code $\code = \Aqn{a}{n}$.}
    
    Since a \qary code of a given size and length with the \acs{MAD} property is unknown, we  formulate the \acs{MAD} code problem, characterize these codes, and show their construction. 
    Moreover, we present a hypergraph representation of codes for \CeRA which facilitates their storage, analysis, and processing. Finally, we introduce a simple method based on the proposed hypergraph representation to perform decoding in \CeRA.%

\subsection{Preliminaries}
    Let $n$ and $q$ be positive integers and let $\Aq{q}=\{0,1,...,q-1\}$ be a finite set, where $q \geq 2$ denotes $|\Aq{q}|$.
    A code $\code$ of block-length $n$ over a $q$-ary alphabet $\Aq{q}$ is a subset of $\Aqn{q}{n}$.
    The elements of $\Aqn{q}{n}$, also called words, are all the ordered $n$-tuples over $\Aq{q}$ of the form $\mathbf{\elementVS}=\elementVS_{n}\elementVS_{n-1}{\cdots}\elementVS_1$, where $\elementVS_{i} \in \Aq{q}$; 
    the elements of $\code$ are called \codeword{s}; 
    and $\M = |\code|$ denotes the size of $\code$. 

    
    The Hamming distance between two $n$-tuples $\mathbf{x},\mathbf{y} \in \Aqn{q}{n}$ is 
    defined as the number of coordinates in which $\mathbf{x}$ and $\mathbf{y}$ differ and it is given by $\hd{\mathbf{x}}{\mathbf{y}}=|\{i:i \in \{1,2,\cdots,n\}, x_{i} \neq y_{i}\}|.$ 
    The average Hamming distance of $\mathbf{\code}$ is calculated as 
    ${\nolinebreak\ahd{\code}=\nicefrac{1}{\M^2}\cdot\sum_{\mathbf{x}\in \code}\sum_{\mathbf{y}\in \code} \hd{\mathbf{x}}{\mathbf{y}}}$ 
    \cite{Fang-Wei_IEEE_TIT-1999_749033}.%

    \subsection{Problem Formulation}
    The \MADCproblem can be stated as follows: given $\Aqn{q}{n}$ of size $\N=|\Aqn{q}{n}|$, select a \RevOne{code} $\code \subseteq \Aqn{q}{n}$ of size $\M$, where $\M \leq \N$, with maximum average distance among all the possible codes of size $\M$ in $\Aqn{q}{n}$.
    Thus, a discrete optimization problem is formally defined as: 
    \begin{maxi}|l|[2]
    {\code \subseteq \Aqn{\ql}{\BigL}}{\ahd{\code}}
    {\label{eq:ahdc_problem_formulation}}{\POne\hspace{0.5cm}}
    \addConstraint{|\code| = \M.}
    \end{maxi}
    An optimal solution for this problem is denominated an \codeacr-\acs{MAD} code. Note that $\POne$ may have more than one possible optimal solution. Since all the devices in \MyCeRA must use the same code, a systematic solution must be implemented by the devices and the \ac{BS}. 
    The problem $\POne$ is the well-known \textit{optimal subset selection problem}, which is an NP-hard problem. 
    Exhaustive search for solving it, however, would be prohibitive, even for moderate instances of the posed problem. 
    %
    Even though there exist numerical methods that can be applied to obtain a solution for $\POne$
    in polynomial time, theoretical results are exploited here for our objective function to change the domain of the problem from the code itself to a general property of the code, namely, the distribution of the \alphabetElement{s} (code symbols) on each \coordinate in $\code$, denoted by $\pi_{i}=\{\xij{i}{j}: j \in \Aq{q}\}$, $i \in \{1,2,\cdots,n$\}.  
    In this way, we are able to obtain an optimal distribution that characterizes the \acs{MAD} codes. Such characterization makes it possible for us to present a systematic solution, thus facilitating the practical design of the scheme.
    
    From \cite{Fang-Wei_IEEE_TIT-1999_749033}, we have $\ahd{\code}$ equals
    \begin{equation}\label{eq:formula_ahd} 
    \Expectation[\hd{X_{\code}}{Y_{\code}}] 
    =\sum^{n}_{i = 1}\left(1-\sum^{q-1}_{j = 0}\piij^2\right)\,,
    \end{equation}
    where $X_{\code}$ and $Y_{\code}$ are two \iid random $n$-tuples with common distribution $P_{\code}=\{P_{\code}(\mathbf{x}): \mathbf{x} \in \Aqn{q}{n}\}$, $P_{\code}(\mathbf{x})$ equals $\nicefrac{1}{\M}, \forall \mathbf{x} \in \code$ and $0$ otherwise;
    $\piij=\sum_{\substack{\mathbf{x} \in \Aqn{\ql}{\BigL} \\ x_{i} = j}} P_{\code}(\mathbf{x})$, and%
    \begin{equation}
    \label{eq:formula_xijpdf}
    \sum_{j=0}^{q-1} \piij=1, \forall { i} = 1,2,{\cdots},n.  
    \end{equation}
    
    By using (\ref{eq:formula_ahd}) and (\ref{eq:formula_xijpdf}), we formulate the following optimization problem:%
    \begin{maxi}|l|[2]
    {\{\piij\}}{\Expectation[\hd{X_{\code}}{Y_{\code}}]}
    {\label{eq:ahdc_reformulated_problem}}{\PTwo\hspace{0.5cm}}
    \addConstraint{(\ref{eq:formula_xijpdf})\,.}
    \end{maxi}
    
    We present next an optimal solution for $\POne$ by constructing a code with the optimal distribution obtained from $\PTwo$. %

    \subsection{The \MyCeRA code: A MAD Code Construction}
    Since the granularity of the \ac{RA} configuration parameters is limited due to signaling overhead\footnote{For instance, even though there are up to $64$ preambles available for \ac{RA} in NR or LTE-M, just $14$ possible values can be configured.}, we restrict ourselves to code sizes of multiple of $\ql$. %
    By assuming this restriction, Theorem \ref{th:equllity_optimalcode_probabilitydistribution} gives the optimal solution for $\PTwo$.
    
    \begin{theorem}[Characterization of MAD codes]
    \label{th:equllity_optimalcode_probabilitydistribution}
    A code $\code \subseteq \Aqn{q}{n}$ of size $\M$ multiple of $q$ is an \codeacr-\acs{MAD} code \myiff $\xij{i}{j}$ equals $\nicefrac{1}{q}$ for all $i=1,2,{\cdots},n$ and $j= 0,1,{\cdots},q-1$.%
    \end{theorem}
    \begin{proof}%
    We prove this theorem by using Fu \etal's theoretical results for the average Hamming distance of a code \cite{Fang-Wei_IEEE_TIT-1999_749033}.
    The objective function in (\ref{eq:ahdc_reformulated_problem}) is upper bounded by 
    $n\left(1-\nicefrac{1}{q}\right)$ \cite[Theorem $1$]{Fang-Wei_IEEE_TIT-1999_749033}, and equality with this upper bound holds \myiff $\xij{i}{j}=\nicefrac{1}{q}, \forall i,j$ \cite[Theorem  $4$]{Fang-Wei_IEEE_TIT-1999_749033}\footnote{Even though the theorem just claims the necessary condition, its proof shows that it is also sufficient.}. Given that $\pi_{i}$ is the distribution of the \alphabetElement{s} on coordinate $i$ over all \codeword{s} in the code, $\xij{i}{j}=\nicefrac{1}{q}, \forall i,j$ implies that the code size is a multiple of $q$ and Theorem \ref{th:equllity_optimalcode_probabilitydistribution} follows.
    \end{proof}
    
    Theorem \ref{th:equllity_optimalcode_probabilitydistribution} gives the distribution $\pi_{i}$ of a \ac{MAD} code. However, a code with that property must be constructed to obtain a solution for $\POne$.
    A specific construction of an \codeacr-\ac{MAD} code is thus provided here, the $(\nc,\ql,\ka)$-\MyCeRA code, where $\ka=\nicefrac{\M}{\ql}$, \RevOne{$\ka \in \{1,2,\cdots,\ql^{n-1}\}$}, is a multiplicative factor of the available preambles of the \MyCeRA scheme. 
    
    The rationale of the proposed  method follows that in  \cite{Sharma_IS-1978}.
    Let $\T$ be the index set of the \codeword{s} in code $\code \subseteq \Aqn{q}{n}$.
    $\T=\{0,1,{\cdots},\M-1\}$, $\M \leq \N = \ql^\nc$. An integer $\kb \in \T$ has a unique representation $\qaryRepresentation$, where $\as_{i} \in \Z_{q}$. Thus, there is a unique correspondence between each $\kb \in \T$ and an $\nc$-tuple of $\Z_{\ql}^{\nc}$ $\aSequence$. We can always arrange the elements of $\Z_{\ql}^{\nc}$ in the natural order of the number they represent.
    Based on this $q$-ary representation of integers $0,1,{\cdots},\M-1$, if $\cSequence$ is the \nth{t} code vector of an $(\nc,\ql,\ka)$-\MyCeRA code, then $\mathbf{\elementC}$ is obtained from the corresponding $\nc$-tuple $\aSequence$ by employing the following mapping: 
    \begin{equation}\label{eq:transformation_mycera}
    \elementC_{i} = \sum^{i}_{l = 1} \as_{l}\pmod{\ql}, \forall i=1,2,\cdots,\nc.
    \end{equation}
    
    As an illustration, Table \ref{table:transformation} gives the above transformation for the $(2,8,k)$-\MyCeRA code, $k=1,2$. Since $\xij{i}{j} = \nicefrac{1}{8}$, $\forall i,j$, the resulting code is a \RevOne{\madcode{2}{8}-MAD} code, assured by Theorem~\ref{th:equllity_optimalcode_probabilitydistribution}. 
    Note that a simple encoding procedure for the \MyCeRA scheme can be derived from this construction. In the beginning of a \superframe, a device randomly selects an integer between $0$ and $M-1$, calculates its $q$-ary representation, and applies the transformation in (\ref{eq:transformation_mycera}) to obtain a random \MyCeRA \codeword, that defines which preamble is transmitted at each \rasubframe, as explained in the beginning of this section.
    \begin{table}[!t]        
    \centering
    \caption{Transformation for the $(2,8,k)$-\MyCeRA code}\label{table:transformation}
    \begingroup 
    \scriptsize 
        \begin{tabular}{cccc|cccc}\hline%
            \bfseries k & \bfseries CW & $\boldsymbol{\Z_{8}^2}$ & \bfseries \MyCeRA & 
            \bfseries k & \bfseries CW & $\boldsymbol{Z_{8}^2}$ & \bfseries \MyCeRA 
            \\\hline\hline
            \multirow{8}{*}{1 or 2} & 0 & 00 & 00 & \multirow{8}{*}{2} & 8 & 10&10\\
            & 1 & 01 & 11 & & 9 & 11 & 21\\
            & 2 & 02 & 22 & & 10 & 12 & 32\\
            & 3 & 03 & 33 & & 11 & 13 & 43\\
            & 4 & 04 & 44 & & 12 & 14 & 54\\
            & 5 & 05 & 55 & & 13 & 15 & 65\\
            & 6 & 06 & 66 & & 14 & 16 & 76\\
            & 7 & 07 & 77 & & 15 & 17 & 07\\
            \hline
        \end{tabular}
    \endgroup
    \vspace{-3mm}
    \end{table}

    \subsection{Hypergraph Representation of Codes for \CeRA Schemes}
    \label{sec:modeling_cera_scheme}
        \RevOne{Let $\rs \leq q$ be the number of preambles used in the code.}\footnote{\RevOne{$\rs$ equals 
        $\rl=\M^{\nicefrac{1}{n}}$
        for \CeRATwo \cite{Vural_CL20187-7908954} and $\ql$ for \MyCeRA.}} Let $\I = \{1,2,{\cdots},n\}$ and $\J = \{0,1,{\cdots},\rs-1\}$, be the index set of the codeword coordinates 
        and code symbols 
        , respectively. We model a code $\code \subseteq \Aqn{\rs}{n}$ of size $\M$ for \CeRA as a simple \mbox{$n$-partite} \mbox{$n$-uniform} hypergraph 
        $\HSymbol = (\V,\E)$, where $\V = \{x_{i,j}: i \in \I \wedge j \in \J\}$ is a finite set of vertices, and  $\E = \{\se_{\kb}: \kb \in \T \}$ %
        is a family of non-empty subsets of $\V$ called hyperedges.
        $\V$ can be partitioned into $n$ disjoint subsets $\V_{1},\V_{2},{\cdots},\V_{n}$ (\mbox{$n$-partite}), %
        and each hyperedge $\se_{\kb} \in \E$ contains $n$ vertices (\mbox{$n$-uniform}), exactly one vertex from each subset of $\V$, i.e. $\se_{\kb} = \{\eSequence\}$, where $\sse_i \in \V_{i}$, $\forall i$. %
        In this model, $\V_{i}=\{x_{i,j}: j \in \J\}$ 
        represents the set of preambles used at \nth{i} \rasubframe, and $\E$ represents the \codeword{s} of $\code$. Thus, there is a hyperedge $\se = \{\sse_1,\sse_2,{\cdots},\sse_n\} \in \E$ \myiff there is a \codeword $\mathbf{c}=c_{n}c_{n-1}{\cdots}c_{1} \in \code$ such that $\sse_{i}=x_{i,c_{i}}$, $\forall i$.

\subsection{Code-expanded RA Decoding}
\label{sec:decoding_cera_codes}   
    Unlike traditional decoding schemes in which the input of the decoding process is a received message and the output is the transmitted \codeword or a list of possible candidates, the decoding process in \CeRA has as input the set of received code symbols (detected preambles) at each \coordinate (\rasubframe) and as output the list of \validcw{s}\footnote{\CeRA can be seen as unsourced \ac{RA} \cite{8006984-Polyanskiy2017} 
    at the \ac{MAC} layer, in which the objective is to communicate a message to get channel access.}, which calls for innovative decoding strategies.
    In this paper, we formulate the \schemetype \ac{RA} decoding as operations over the hypergraph representation of the code, as given in Lemma \ref{th:lemma_mad_code_decoding}.
    Note that Lemma \ref{th:lemma_mad_code_decoding} applies to any existing code for \schemetype \ac{RA}.
    \begin{definition}[Set of \validcw{s}]
    \label{def:valid_codewords}%
    Given $\Y= \{{\Y}_{i}: i \in \I \} \subseteq \V$, 
    where ${\Y}_{i}$ is the set of preambles detected at \nth{i} \rasubframe, the set of valid codewords inferred by the \ac{BS} in a \superframe 
    is given by all \codeword{s} of the form $\cCodeword \in \code$ such that $x_{i,\elementC_{i}} \in \Y_{i}$, $\forall i$. 
    \end{definition}
    \begin{definition}[Induced subhypergraph]
    \label{def:induced_subhypergraph}%
    An \textit{induced subhypergraph} of a hypergraph \mbox{$\HSymbol = (\V,\E)$}, is the hypergraph $\IGraph{\Y}{\E\textprime}$, where $\Y {\subseteq} \V$ and a hyperedge $\se \in \E$ is in $\E\textprime$ \myiff all the vertices of $\se$ are in $\Y$. 
    \end{definition}
    \begin{lemma}[Decoding in \CeRA] \label{th:lemma_mad_code_decoding}
    Given a code $\code \subseteq \Aqn{r}{n}$ with $\HSymbol = (\V,\E)$ and $\Y= \{{\Y}_{i}: i \in \I \} \subseteq \V$, where ${\Y}_{i}$ is the set of preambles detected at \nth{i} \rasubframe of a given \superframe, the set of valid \codeword{s} inferred in the \superframe is given by the hyperedge set of the subhypergraph induced by $\Y$. 
    \end{lemma}
    \begin{proof} 
        Assume $\HSymbol = (\V,\E)$ of a code $\code \subseteq \Aqn{r}{n}$. 
        Let $\Y= \{{\Y}_{i}: i \in \I \} \subseteq \V$ be the set of preambles detected at each \rasubframe of a \superframe. 
        From Definition \ref{def:induced_subhypergraph} and the model described in Section \ref{sec:modeling_cera_scheme}, there is a hyperedge $\se= \{\sse_1,\sse_2,{\cdots},\sse_n\} \in \E$ in $\IHSymbol{\Y}$ \myiff $\sse_i \in \Y_{i}$, $\forall i$. 
        Let us call such a hyperedge a \coverededge and the represented \codeword 
        $\cCodeword \in \code$ such that $\sse_{i}=x_{i,c_{i}}$, $\forall i$, a \coveredcw.
        Thus, using Definition \ref{def:valid_codewords} and the above argument, the set of \validcw{s} is equal to the set of \coveredcw{s}, and Lemma \ref{th:lemma_mad_code_decoding} follows.%
    \end{proof}

\subsection{Example}\label{sec:example}
    \RevOne{The following example illustrates the advantage of using \MyCeRA (Fig. \ref{fig:codeambiguity_mycera}) when compared to \CeRA (Fig. \ref{fig:codeambiguity_cera})  \cite{Thomsen_TETT2013,Vural_CL20187-7908954}. 
    It can be noted that although the two schemes have the same number of available \codeword{s} (\cBallsCodebook \xspace balls) and transmitted \codeword{s} (\cBallsReceived \xspace balls), the number of \codeword{s} inferred (\cBallsValid \xspace balls) by \MyCeRA is less
    than that inferred by \CeRA, which evinces the greater efficiency in resource utilization achieved by the proposed scheme.}

\section{Analytical Model for \RevOne{CeRA Schemes}}\label{sec:analytical_model} 
    We introduce an analytical model for \CeRA schemes which covers both the \CeRATwo scheme in \cite{Vural_CL20187-7908954} and the proposed \MyCeRA scheme. Given that the impact of the \codeambiguityproblem on resource utilization efficiency has been neglected in the literature, we include not only the \ac{RA} success probability, but also \resourceefficiency in the model.
    
    Assume that $\Nb$ \device{s} employ a \CeRA scheme with $n$ \rasubframe{s} per \superframe. Let $\rvX$ denote the \ac{RV} of the number of \device{s} contending per \codeword in a given \superframe. The probability distribution of $\rvX$ follows a binomial distribution with parameters $\Nb$ and $\pb=\nicefrac{1}{\M}$, where $\M$ is the number of available \codeword{s} \cite{Thomsen_TETT2013}.
    Thus, $\Probability(\rvX=m) = \Bknp{m}{\Nb}{\pb} = \binom{\Nb}{m}\pb^m\left(1-\pb\right)^{\Nb-m}$. 
    Note that $\M$ equals $\ql {\cdot} \ka$ for the \MyCeRA scheme and $\rl^\BigL$ for the \CeRATwo scheme, \RevOne{where $\rl \in \{1,2,\cdots,\ql\}$}. 
    The expected number of \codeword{s} chosen by at least one \device\xspace \RevOne{(selected)} and by a single \device\xspace \RevOne{(non-collided)} are given by
    $\NC=\M\cdot[1-\Bknp{0}{\Nb}{\pb}]=\M\cdot[1-\left(1-\pb\right)^{\Nb}]$
    and
    $\NS=\M\cdot\Bknp{1}{\Nb}{\pb}=\Nb\left(1-\pb\right)^{\Nb-1}$%
    , respectively.

    The \RevOne{\codeword non-colission} probability for \CeRA schemes is 
    $P_{N}=\nicefrac{\NS}{\K}=\left(1-\pb\right)^{\Nb-1}$,
    and the \ac{RA} success probability can be calculated as $P_{S}=P_{N} \cdot \Palloc$ \cite{Vural_CL20187-7908954}\footnote{\RevOne{We consider a device to be successful if its \codeword is served and non-collided. Since $K$ is the number of contending devices perceived by the \ac{BS} (see Section \ref{sec:practial_example}), the codeword detection probability is not included.}}, where $\Palloc$ is the resource allocation probability.  
    Let $\rvV$ be the \ac{RV} of the number of \validcw{s}. Since $\rvV$ can be greater than \RevOne{the number of available resources in a superframe} ($\R$), not all the \validcw{s} can be served in a given \superframe. Thus, $\Palloc$ equals $\nicefrac{\R}{\Expectation[\rvV]}$ if $\Expectation[\rvV] > \R$, or equals $1$ otherwise \cite{Vural_CL20187-7908954}. Theorem \ref{th:e_v_formula} gives $\Expectation[\rvV]$ for the case when $n=2$, which is of main interest due to the practical reasons \RevOne{in Section \ref{sec:practial_example}}. 
    The proof of Theorem \ref{th:e_v_formula} is described in a way that, given the code, $\Expectation[\rvV]$ can be calculated numerically for a general $n$ value. %

\begin{figure*}[!t]\vspace{-2.1em}\centering\hspace{\PIHSpaceFigures}%
    \subfloat[RA success probability ($P_{S}$)]{\label{fig:mycera_ps}
    \includegraphics[width=\PRatioWidthFigures]{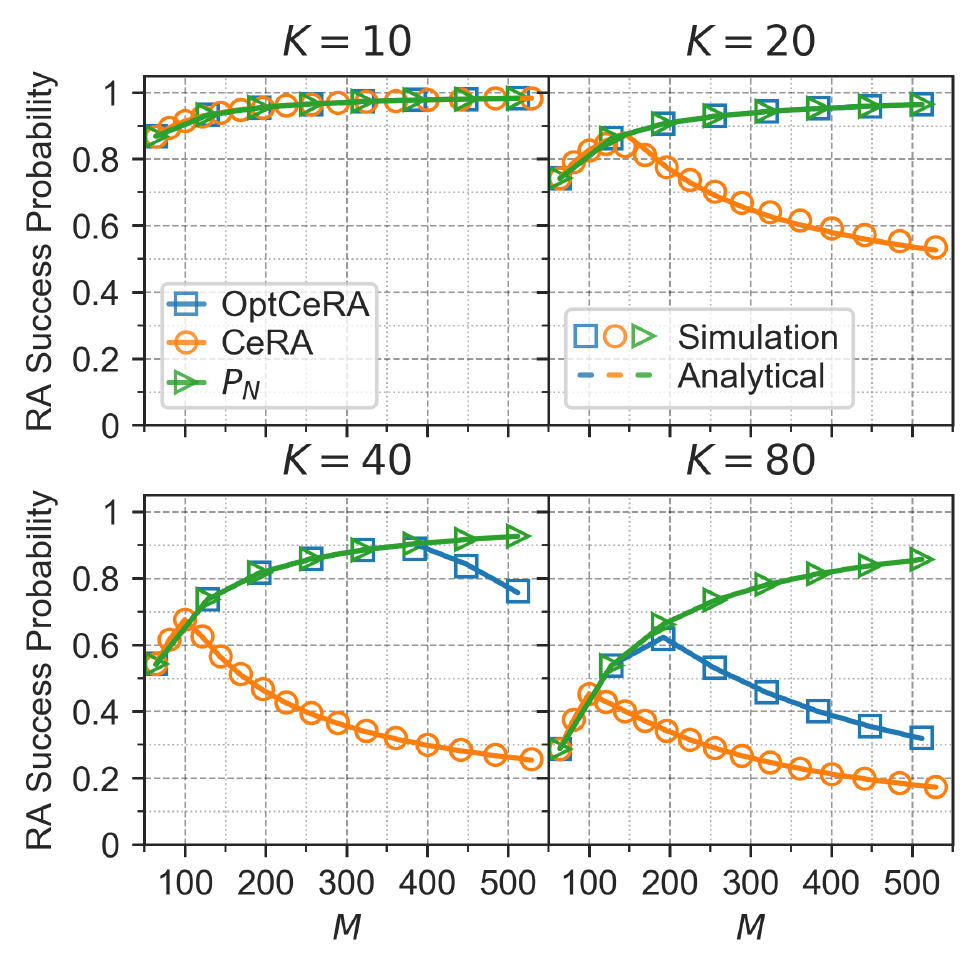}}\hspace{\PHSpaceFigures}
    \subfloat[Number of inferred valid codewords
    ({$\Expectation[\rvV]$})]{\label{fig:mycera_valid} 
    \includegraphics[width=\PRatioWidthFigures]{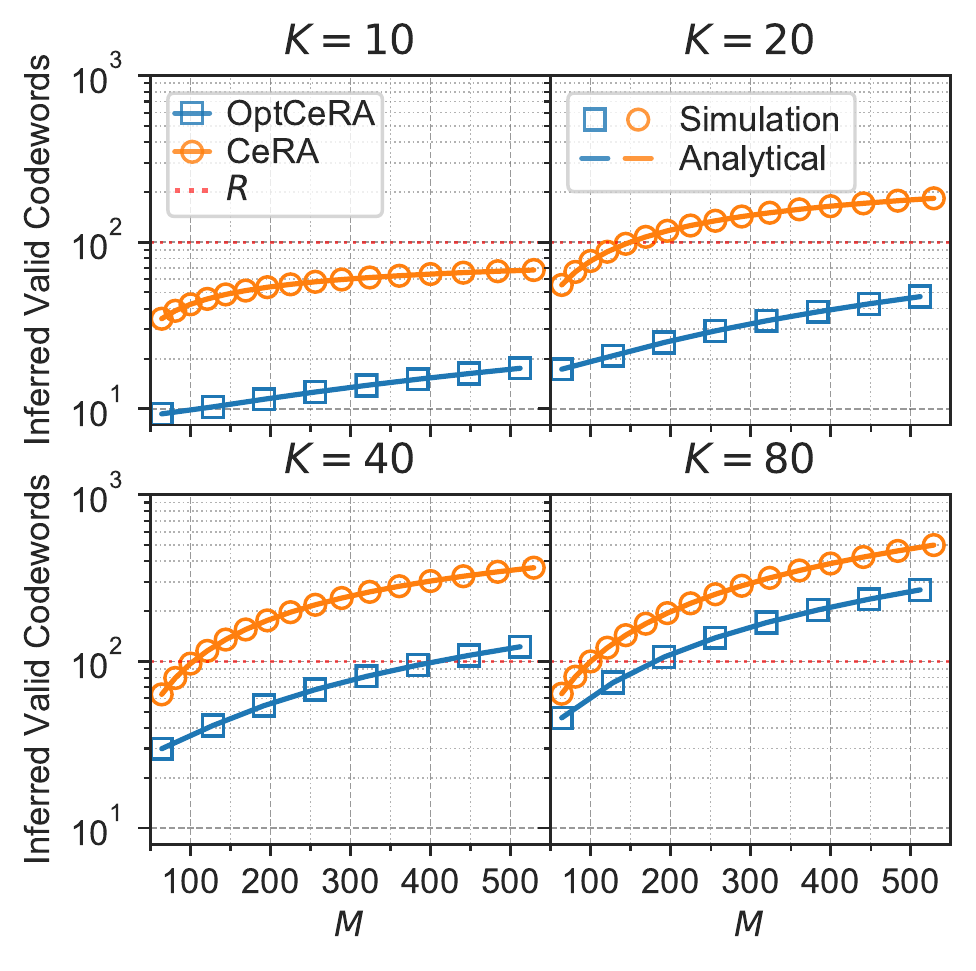}}\hspace{\PHSpaceFigures}
    \subfloat[Grant utilization ({$\Expectation[\rvU]$})]{\label{fig:mycera_putil} 
    \includegraphics[width=\PRatioWidthFigures]{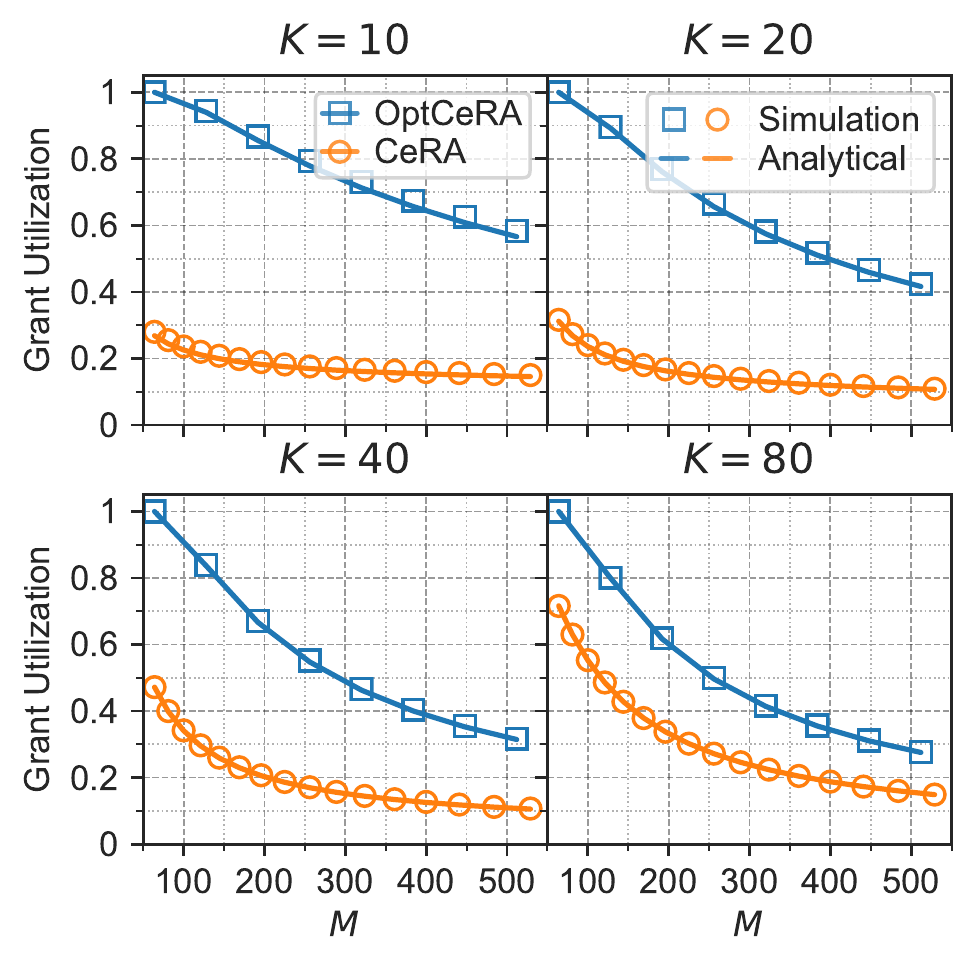}}%
\caption{Performance evaluation results}
\vspace{-4mm}
\label{fig:performance_evaluation_nonadaptive}
\end{figure*}    
   
    \begin{theorem}[Expectation of $\rvV$]
    \label{th:e_v_formula} 
    If $\Nb$ devices are trying \ac{RA} simultaneously employing a \schemetype \ac{RA} scheme with code $\code \subseteq \Aqn{\rs}{2}$ of size $\M$, $\rs \leq \ql$, and $\xij{i}{j}$ equals $\nicefrac{1}{\rs}$, for all $i=1,2$ and $j = 0,1,{\cdots},\rs-1$, then $\Expectation\left[\rvV\right]$ is calculated as
    \begin{equation}
    \label{eq:ev_formula}
    \Expectation[\rvV] = \fExpectationVNTwo. 
    \end{equation}
    \end{theorem}
    \begin{proof}
    Consider a network employing a \schemetype \ac{RA} scheme. Suppose each of $\Nb$ devices trying \ac{RA} in a given \superframe randomly selects a \codeword from a code $\code \subseteq \Aqn{\rs}{n}$ of size $\M$. 
    Let $\HSymbol = (\V,\E)$ be the hypergraph representation of $\code$. 
    In the \superframe, the base station detects a set of preambles ${\Y} = \{\Y_{i}: i \in \I\}$. 
    By Lemma \ref{th:lemma_mad_code_decoding}, the set of \validcw{s} 
    can be obtained from the set of hyperedges of $\IHSymbol{\Y}$. 
    %
    %
    Let us call an hyperedge \textit{covered} if it is in $\IHSymbol{\Y}$ and \textit{non-covered} otherwise. 
    Let $A_{\kb}$ be the event that $\se_{\kb} \in \E$ is \textit{covered} and $A^{c}_{\kb}$ its complement; let $I_{A^{c}_{\kb}}$ be the indicator function of $A^{c}_{\kb}$.
    Let $\rvV = \sum_{\kb \in \T} 1-I_{A^{c}_{\kb}}$ be the \ac{RV} of the number of \coverededge{s} or, equivalently, the number of \validcw{s}, and $\Expectation[\rvV] = \sum_{\kb \in \T} 1-\Probability(A^{c}_{\kb})$  the expectation of $\rvV$, where $\Probability(A^{c}_{\kb})$ is the probability of $A^{c}_{\kb}$.
    Let $A^{c}_{\kb,i}$ be the event that $\sse_{i} \in \se_{\kb}$ is not in $\Y_{i}$, i.e. the \nth{i} coordinate of the \nth{t} \codeword is not in the set of preambles detected at \nth{i} \rasubframe. 
    Since an element $\se_{\kb} \in \E$ is in $\IHSymbol{\Y}$ \myiff all the vertices of $\se_{\kb}$ are in $\Y$, $A_{\kb}=\bigcap_{i=1}^{n} A_{\kb,i}$ and $A^{c}_{\kb}=\bigcup_{i=1}^{n} A^{c}_{\kb,i}$. By applying the inclusion-exclusion principle, 
    $%
    \Probability\left(A^{c}_{\kb}=\bigcup_{i=1}^{n} A^{c}_{\kb,i}\right)= \sum_{\emptyset \neq \Q\subseteq\{1,2,\cdots,n\}}  (-1)^{|\Q|-1} \Probability\left(\bigcap_{i \in \Q} A^{c}_{\kb,i}\right).
    $
    For the case $n=2$, 
    $\Probability(A^{c}_{\kb})=\Probability\left(A^{c}_{\kb,1}\right)+\Probability\left(A^{c}_{\kb,2}\right)-\Probability\left(A^{c}_{\kb,1}{\cap}A^{c}_{\kb,2}\right)$.  Moreover, if $\xij{i}{j} = \nicefrac{1}{\rs}$, $\forall i,j$
    , then, we have $M =\rs \cdot \ka$, and, $\forall \kb,i$,
    $\Probability\left(A^{c}_{\kb,i}\right)=\left(1-\nicefrac{\ka}{\M}\right)^{\Nb}$ 
    and $\Probability\left(
    A^{c}_{\kb,1}{\cap}A^{c}_{\kb,2}\right)=\left(1-\nicefrac{(2\ka-1)}{\M}\right)^{\Nb}$.
    Note that $\ka$ and $2\ka-1$ are the number of \codeword{s} in $\code$ that satisfy the complement of the event for which probability has been calculated by considering a single device. 
    Theorem \ref{th:e_v_formula} follows after inserting  $\Probability(A^{c}_{\kb})$ into $\Expectation[\rvV]$, summing up over all $\kb \in \T$, and substituting $k$ by $\nicefrac{\M}{\rs}$.
    \end{proof}
    
    Finally, let $\rvU$ be the \resourceefficiency defined as 
    $\rvU=\nicefrac{\rvA}{\rvG},$
    where $\rvA$ is the number of grants actually used and $\rvG$ is the number of uplink grants issued for RA devices in a \superframe;
    $\Expectation[\rvU]=\Expectation[\nicefrac{\rvA}{\rvG}]=\Expectation[\nicefrac{\rvC\cdot\Palloc}{\rvV\cdot\Palloc}]%
    \approx \nicefrac{\Expectation[\rvC]}{\Expectation[\rvV]},$ 
    where $\rvC$ is the number of grants actually needed. Thus, $\Expectation[\rvC]$ equals $\NC$ for both the \CeRATwo and  \MyCeRA schemes. 
    
    \RevOne{
    \section{Practical Considerations}%
    \label{sec:practial_example}
    The  \MyCeRA scheme can be applied to existing technologies, e.g. 5G NR, NB-IoT, LTE-M, LTE-A, and LTE, in a straightforward way. All these technologies use  system information broadcast (SIB) messages to dynamically configure the RA channel (RACH) at some subframes (\rasubframe{s}) by setting the value of  the PRACH configuration index parameter as well as  allocating $q$ preambles for contention-based \ac{RA}. 
    
    The values of $n$, $q$, and $k$ need to be set to configure the \MyCeRA scheme, but caution needs to be taken since both energy consumption and access delay may increase with the  value of $n$. When $n=2$, up to $64^2$ \codeword{s} can be provided which is greater than the available control/data resources in the aforementioned technologies. 
    Thus, the value of $n$ can be fixed to $2$. Moreover, this also helps to reduce the signaling overhead. 
    While the value of $n$ can be preconfigured, the value of $k$  needs to be added to an SIB message. 
    
    For using the analytical model in Section \ref{sec:analytical_model}, the value of $K$ is needed, which can be estimated at the \ac{BS} by using existing approaches based on idle and detected preambles. In this way, $K$ can be regarded as the number of devices contending with a codeword whose preambles are detected in a superframe.}%

\section{Numerical Results}
\label{sec:Results}
    This section compares the \CeRATwo \cite{Vural_CL20187-7908954} (hereinafter referred to as \CeRA) and \MyCeRA schemes. 
    The proposed analytical model is also validated by comparing the analytical and simulation results. 
    Results are presented as a function of $\M$\footnote{\RevOne{$\M$ equals $\rl^\BigL$ and $\ql {\cdot} \ka$ for the \CeRA and \MyCeRA schemes, respectively. To assign values to $\M$ in Fig. \ref{fig:performance_evaluation_nonadaptive}, we used $\rl \in \{8,9,\cdots,23\}$ and $\ka \in \{1,2,\cdots,8\}$.}} and $\Nb$ for $q=64$, $n=2$, and $\R= 100$.\footnote{Considering an LTE-M network with a 2-\subframe \rasuperframe every $10$ \subframe{s}, and $6$ uplink resources available per \subframe not being a \rasubframe, $\R=(10-2){\cdot}6=48$. Moreover, considering a \ac{NOMA} technique with  100\%-overloading \cite{Liang_7929338}, we obtain roughly $100$ uplink resources per \rasuperframe, which is a realistic value for $\R$ in a \ac{CIoT} network. Higher values of $R$, however, are possible by decreasing the frequency of the RA subframes.}
    The proposed \MyCeRA scheme outperforms the \CeRA scheme for all device loads and code sizes considering $P_{S}$ (Fig.~\ref{fig:mycera_ps}), $\Expectation[\rvV]$ (Fig.~\ref{fig:mycera_valid}) and $\Expectation[\rvU]$ (Fig.~\ref{fig:mycera_putil}), and achieves more than $50\%$, $100\%$ and $300\%$ performance gain, respectively. 
    
    When $\Palloc^{\CeRA}=1$ ($\Expectation[\rvV] \leq \R$), the two schemes produce the same $P_{S}$ values (Fig. \ref{fig:mycera_ps}), but the granted resources are utilized very inefficiently when the \CeRA scheme is employed (Fig.~\ref{fig:mycera_putil}). 
    Moreover, a cellular network typically has two groups of devices: random access-based devices and scheduling-based devices (e.g. based on semi-persistent or dynamic scheduling). Hence, the low grant utilization of the \CeRA scheme may significantly affect the performance of both groups of \device{s} in a coexistence scenario.

    Furthermore, when $\Palloc<1$ ($\Expectation[\rvV] > \R$) in each scheme, its $P_{S}$ value decreases as $\M$ increases even though $P_{N}$ increases.
    This situation occurs because the number of \validcw{s} is larger than the available resources (Fig.~\ref{fig:mycera_valid}), forcing the \network to randomly allocate the available resources among the \validcw{s}. Thus, some allocated resources are not actually used by any device while a certain portion of the devices waiting for an uplink grant does not receive response, even though their preamble transmission are correctly received at each \rasubframe. 
    This may strongly impact on the performance of \schemetype \ac{RA} in a limited resource regime. 
    However, for a given number of devices, this happens with a higher value of $\M$ when \MyCeRA scheme is used, allowing \MyCeRA scheme to support higher $P_{S}$ values than does the \CeRA scheme. 
    
    All of these gains are achieved thanks to the proposed optimized \MyCeRA code. It significantly reduces the number of \codeword{s} that the \network infers, effectively alleviating the \codeambiguityproblem. Moreover, these results show that the \MyCeRA scheme, as well as \CeRA and other RA schemes (\cite{Thomsen_TETT2013,Vural_CL20187-7908954}), have an optimal point of operation which defines the system parameter values (e.g. $\M$) and depends on the device load ($\Nb$). Such an optimal point can be derived by using the analytical model  proposed here, which is in full agreement with Monte Carlo simulation results ($10^5$ iterations).

\section{Conclusion}
\label{sec:Conclusion}
\RevOne{This letter has addressed the \codeambiguityproblem in \schemetype \ac{RA} for \ac{mMTC} by allowing devices to select codewords from a novel code with maximum average distance. 
The proposed scheme reduces the number of valid codewords that can be inferred, and greatly increases the \ac{RA} success probability as well as the efficiency in resource utilization. 
The proposed maximum average distance code is likely to find other innovative applications.}
\bibliographystyle{IEEEtran}
\bibliography{Fonseca_WCL2021-0549}
\end{document}